\documentclass[letterpaper]{llncs}

\usepackage{graphicx}
\usepackage{latexsym}
\usepackage{verbatim}
\usepackage{xspace}
\usepackage{amsmath}
\usepackage{amsfonts}
\usepackage{amssymb}
\usepackage[utf8]{inputenc}
\usepackage{color}
\usepackage{multirow}
\usepackage{tabularx}

\newlength{\ddbitwidth}
\newcommand{\Z}{%
  \makebox[\ddbitwidth]{\texttt{0}}%
}
\newcommand{\U}{%
  \makebox[\ddbitwidth]{\texttt{1}}%
}
\newcommand{\D}{%
  \makebox[\ddbitwidth]{\texttt{\_}}%
}
\newcommand{\B}{%
  \makebox[\ddbitwidth]{}%
}
\definecolor{ddgrey}{rgb}{0.6,0.6,0.6}
\newcommand{\ddgreymid}{%
  \makebox[\ddbitwidth]{%
    \raisebox{0pt}[\height][0pt]{%
      \textcolor{ddgrey}{\ensuremath{\mid}}%
}}}
\newcommand{\ddstrtimes}[2]{%
  #1$\times$\texttt{#2}%
}
\newlength{\ddlenRes}
\newcommand{\ddResTtlSq}[1]{\resizebox{\ddlenRes}{\height}{#1}}
\def\ddTT#1=#2Z{%
  \ensuremath{#1}\mbox{}\ensuremath{=}\mbox{}\ensuremath{#2}%
}
\newcommand{\ddT}[1]{%
  \ddResTtlSq{\ddTT#1Z}%
}
\newcommand{\ddTR}[4]{%
  \hline#1&\ddResTtlSq{#2}&\ddResTtlSq{#3}&\ddResTtlSq{#4}&\ddTRX%
}
\newcommand{\ddTRX}[5]{%
  \ddT{#1}&\ddT{#2}&\ddT{#3}&\ddT{#4}&\ddT{#5}&\ddTRXX%
}
\newcommand{\ddTRXX}[7]{%
  \ddT{#1}&\ddT{#2}&\ddT{#3}&\ddT{#4}&\ddT{#5}&\ddT{#6}&\ddT{#7}\\\hline\hline%
}
\newcommand{\ddRR}[4]{%
  \texttt{#1}&#2&#3&#4&\ddRRX%
}
\newcommand{\ddRRX}[5]{%
  #1&#2&#3&#4&#5&\ddRRXX%
}
\newcommand{\ddRRXX}[7]{%
  #1&#2&#3&#4&#5&#6&#7\\\hline
}
\newlength{\ddwidthtt}
\newcommand{\aZ}{%
  \makebox[2\ddwidthtt]{\texttt{0}}%
}
\newcommand{\aU}{%
  \makebox[2\ddwidthtt]{\texttt{1}}%
}

\newcommand{\oA}{%
  \makebox[2\ddwidthtt]{\underline{\texttt{\ }}}%
}
\newsavebox{\ddboxI}
\newsavebox{\DDboxII}

\begin{document}

\newlength{\halftextwidth}
\setlength{\halftextwidth}{0.47\textwidth}
\def\halffigsize{2.2in}
\def\thirdfigsize{1.5in}
\def\negvspace{0in}
\def\posvspace{0em}

\newtheorem{mydefinition}{Definition}
\newtheorem{myexample}{Example}
\newtheorem{myobservation}{Observation}
\newtheorem{mytheorem}{Theorem}
\newtheorem{mylemma}{Lemma}
\newtheorem{mycor}{Corollary}
\newenvironment{runexample}{{\bf Running example:} \it}{\rm}
\newcommand{\myprroof}{\noindent {\bf Proof:\ \ }}
\newcommand{\myqed}{\mbox{$\Box$}}

\newcommand{\constraint}[1]{\mbox{\sc #1}\xspace}
\newcommand{\element}{\constraint{Element}}

\newcommand{\mymax}{\mbox{\rm max}}
\newcommand{\mymin}{\mbox{\rm min}}
\newcommand{\mylog}{\mbox{\rm log}}

\newcommand{\todo}[1]{{\tt (... #1 ...)}}

\title{Finding Synchronization Codes to Boost Compression by
  Substring Enumeration}
\author{
Dany Vohl,
Claude-Guy Quimper,
\and
Danny Dub\'{e}}
\institute{Universit\'{e} Laval, Canada\\
\texttt{Dany.Vohl.1@ulaval.ca}\\
\texttt{Claude-Guy.Quimper@ift.ulaval.ca}\\
\texttt{Danny.Dube@ift.ulaval.ca}}

\maketitle
\begin{abstract}
  Synchronization codes are frequently used in numerical data
  transmission and storage.  Compression by Substring Enumeration
  (CSE) is a new lossless compression scheme that has turned into a
  new and unusual application for synchronization codes.  CSE is an
  inherently bit-oriented technique.  However, since the usual
  benchmark files are all byte-oriented, CSE incurred a penalty due to
  a problem called \emph{phase unawareness}.  Subsequent work showed
  that inserting a synchronization code inside the data
  \emph{before} compressing it improves the compression performance.
  In this paper, we present two constraint models that compute the
  shortest synchronization codes, i.e. those that add the fewest
  synchronization bits to the original data.  We find synchronization codes for blocks of
  up to 64 bits.
\end{abstract}

\section{Introduction}

Synchronization codes are frequently used in numerical data
transmission and
storage~\cite{bruyere97,do02,golomb65,scholtz66,wijngaarden01,stiffler71}.
While they might not be as well known as error-correction and
error-detection codes, they still play a crucial role.  Indeed,
whenever the reception of data becomes ill-synchronized, data gets
distorted.  In such a case, even error-correction codes are of
no help since they might interpret some data symbols as control
symbols and vice-versa.  In almost every application, keeping the
transmission and storage of data correctly synchronized is considered
to be a separate, lower-level task, which is
handled by
synchronization codes.

Recent work on data compression gives synchronization codes a new
and rather unusual purpose~\cite{dube10b,dube11}.  These are
used to \emph{boost} the performance of a specific data compressor.  A
preprocessing step is added to the data compressor which
consists in inserting synchronization codes to the data.  Since
inserting a synchronization code causes the data to expand, it seems
at first glance that this preprocessing is counter-productive.
However, in this specific setup, the insertion of the synchronization
codes improves the performance of the subsequent compression step, as
measured (in an absolute sense) on the compressed data.\footnote{That
  is, the size of the synchronized-and-compressed data is smaller than
  the size of the directly compressed data, in an absolute sense.  We
  insist on the term ``absolute'' here since making synchronized data
  more compressible, in a \emph{relative} sense, is trivial because
  one only needs to stuff the data with highly repetitive paddings.}
The data compression technique is called \emph{Compression by
  Substring Enumeration} (CSE)~\cite{dube10a}.

In order to be successful as compression boosters, the considered
synchronization codes must obey some properties.  In particular, it is
desirable for the codes not to cause too big an expansion of the data.
Indeed, the subsequent compression phase has to compensate for all of
this expansion.  Also, the important characteristics of the
synchronization codes is their synchronization power; i.e.~how
effectively they provide synchronization information.  In the light of
these goals, this paper aims at designing strong synchronization codes
that expand data as little as possible.

The paper is structured as follows.  Section~\ref{sect:syncodes}
briefly surveys the synchronization codes and formally defines the
family of codes that we focus on.  Section~\ref{sect:CSE} presents the
key features of CSE and gives some intuitive reasons why CSE has the
potential to be boosted by synchronization codes.  Next, we present
two constraint models that we use to compute a synchronization code.
Section~\ref{sect:constraint_model} presents a model that is based on
constraint programming.  Section~\ref{sect:pseudo_boolean_model}
presents a pseudo-Boolean model.  Section~\ref{sect:symmetries}
discusses about ways to break symmetries among the sets of codes that
are considered in order to reduce the time spent by the solvers in
searching.  The experiments described in
Section~\ref{sect:experiments} produced synchronization codes for
words of~2~to~8, 16, 32, and 64~bits.  Our codes are proved optimal for words of
2 to 8, 16 and 32~bits.


\section{Synchronization Codes}
\label{sect:syncodes}

\subsection{Overview}

There exists a wide variety of synchronization codes.  Codes differ
according to the properties they feature and the principles they are
based on.  Naturally, their effectiveness depends on the application
at hand.  We start by illustrating the need for synchronization codes
by considering two examples of applications.

Our first application is a serial communication link.  In serial
communication, symbols (typically bits) are transmitted one after the
other through a channel from an emitter to a receiver.  When the pace
of transmission is controlled by a clock, it might be the case that,
in fact, the emitter and the receiver each have their own separate
clock.  In such a context, it is likely that the clocks do not have
exactly the same frequencies.  If the receiver's clock is slower than
the emitter's clock, then bits might go undetected from time to time
by the receiver.  On the other hand, if the receiver's clock is
faster, then a single bit sent by the emitter might get sampled twice
by the receiver.  Such communication errors are called synchronization
errors.  Naturally, some form of synchronization mechanism has to be
provided in order to prevent these errors.  In some contexts, the
designer of the communication device might have the luxury to join a
control channel to the data channel (e.g.,\ on the motherboard of a
computer).  The control channel transmits synchronization symbols
while the data channel transmits the payload symbols.  In other
contexts, no separate control channel is available and the
synchronization information must be inserted within the payload data.

Our second application is a hard disk.  In a hard disk, the read/write
(RW) head hovers above one of the tracks of the spinning surface.  The
RW process of the hard disk is subject to the same synchronization
errors as the serial communication link.  For instance, serial-like
synchronization errors are possible due to the imprecision in the
rotation speed of the disk or some other factor.  But the RW process
also faces additional challenges.  Indeed, the RW head has to
determine where a track starts; in particular, after it moves from
one track to another.  Some sort of marker has to tell the RW head where the
track starts (or, at a finer level, where the individual sectors of
the track start).  Nowadays, the magnetic bits are so densely packed
on the tracks that it is unrealistic to rely on a physical device to
mark the beginning of the tracks.  Instead, markers identifying the
start of the tracks (or sectors) have to be integrated among the
magnetic information.  Since waiting for a unique ``track beginning''
marker would cause a waste of time, we are likely to require the
inscription of more frequent ``sector beginning'' markers.  These
enable quicker recovery of the synchronization and also offer the
opportunity to join a header that includes
additional useful information such as the sector IDs.

\subsection{Characteristics of Synchronization Codes}

Below is a list of considerations
(either requirements or features) about
synchronization codes.  But before we give the list, we introduce a
few definitions and a (not completely general) notation for
synchronization codes that we use throughout the paper.

A sequence $A = a_0a_1 \ldots a_{|A|-1}$, is an ordered list of
characters taken from an alphabet $\Sigma$.  The length of a sequence
$A$ is the number of characters in the sequence and is denoted $|A|$.
The subsequence $A[i..j]$ is formed by the characters $a_i a_{i+1}
\ldots a_{j-1} a_j$.  The \emph{rotation by one}\footnote{A complete
  name for the operation ought to be the \emph{rotation to the left by
    one}.  However, this name would be unnecessarily long as we never
  consider rotations to the right.} of a sequence is obtained by
deleting the first character of the sequence and appending it to the end
of the sequence.  For instance, the rotation by one of the sequence
$\mathtt{abcd}$ is $\mathtt{bcda}$.  We denote by $A^1$ the rotation
by one of a sequence $A$.  A \emph{rotation by $i$} of a sequence $A$,
which we denote by~$A^i$, is obtained by applying $i$~times a rotation by
one.

We use the following notation to specify a synchronization code (at
least, to specify the family of synchronization codes we focus on).  A
synchronization code is specified as a sequence taken from the
alphabet $\Sigma \cup \{ \D \}$.  The special character~\D\ acts as a
wildcard.  We say that two symbols \emph{match} each other if they are
identical or if at least one of them is a wildcard.  Sequence~$A$
matches sequence~$B$ if $|A| = |B|$ and if the character $a_i$
matches the character $b_i$ for all $0 \leq i < |A|$.  For instance,
the sequence $\mathtt{a\ \_\ b\ \_\ \_\ c}$ matches the sequence
$\mathtt{a\ b\ \_\ c\ \_\ c}$.  Let~$C$ be a synchronization code
and~$d$ be the number of wildcards in~$C$.  We say that a
\emph{clear-text} sequence~$D$ of $d$~symbols can be encoded using a
synchronization code~$C$ to obtain a \emph{synchronized}
sequence~$D_C$.  One obtains $D_C$ by substituting the first wildcard
in~$C$ by~$d_0$, the second wildcard in $C$ by $d_1$, and so on.
When~$|D|$ is a multiple of~$d$, the successive $d$-symbol blocks
of~$D$ are encoded using the same process.  For instance, by encoding the
original sequence $D=\Z\U\U\U$ using the synchronization code
$C=\D\D\U\U\Z$, we obtain the encoded sequence
$D_C=\underline{\Z\U}\U\U\Z\underline{\U\U}\U\U\Z$ (for the sake of
clarity, the data symbols are underlined).

Here is a list of considerations about synchronization codes.

\begin{itemize}

\item \textbf{Existence of separate channels for control and data.}
  In certain transmission configurations, there are two channels
  available, one for data and one for control, while in other
  configurations, there is only one channel in which both
  data and control symbols are transmitted.  When there is a single
  channel and the distinction is possible, we refer to the
  symbols dedicated to control as \emph{control symbols} and to the
  symbols dedicated to the transmission of pure data as \emph{data
  symbols} or \emph{payload symbols}.

\item \textbf{Transmission through time and/or space.}  Transmission
  through space refers to transmission in the usual sense, i.e.~from
  one point to the other.  Transmission through time refers to
  \emph{storage}.  That is, writing symbols on a storage device can be
  seen as the emission of the symbols and later reading the symbols
  from the device, as the reception of the symbols.  For example, in
  the case of a hard disk, the emitter and the receiver are the same
  RW head, but at two different points in time.

\item \textbf{Presence of a feedback link.}
  The synchronization mechanism may have access to a feedback link.
  For example, the TCP/IP protocol uses feedback
  to control the integrity of the transmission process~\cite{tcp89}.
  On the other hand, a hard disk RW head, when reading a track, might
  read something that it has written long ago.
  In a sense, when it reads, the RW completes a transmission through
  time.  If it discovers a synchronization error, the RW head has
  no means to ask itself to reemit the symbols.

\item \textbf{Size of the alphabet.}
  The considered alphabet is either binary or larger.

\item \textbf{Strength of the synchronization.}  A synchronization
  code provides either hard or soft synchronization.  By \emph{hard}
  synchronization, we mean that, whenever an \emph{a priori}
  ill-synchronized receiver obtains a sufficient number of error-free
  symbols, it is guaranteed to be able to recover synchronization.  We
  also say \emph{reliable} synchronization.  \emph{Soft}
  synchronization means that, in the worst case, it might remain
  impossible to identify the position inside of the data with
  certainty, even after receiving an arbitrary number of error-free
  symbols.

\item \textbf{Blockiness of the synchronization.}
  A synchronization code works either in a blockwise
  or in a continuous fashion.

\item \textbf{Size of the blocks.}
  If a synchronization code works in a blockwise fashion,
  it handles either fixed-size or variable-size blocks.

\item \textbf{Variability of the control-symbol pattern.}
  A synchronization code uses either a fixed or a variable
  pattern of control symbols.  For example, \D\D\Z\U\U\ inserts three
  constant bits.  A variable pattern would use three functions
  to determine the value of the three control bits, based on the data
  bits.  (Note that a variable pattern is not
  representable using our notation.)  Strictly speaking, a fixed
  pattern is a special case of a variable pattern.

\item \textbf{Position of the control bits.}
  If it works in a blockwise fashion, a synchronization code either
  groups the control symbols in a header or intersperses them
  among the data symbols.

\item \textbf{Invariance of the payload symbols.}
  If it works in a blockwise fashion, a synchronization code either
  alters or preserves the payload symbols.  A typical
  case where a synchronization code alters the payload data is when
  there is a header that contains a unique signature that marks the
  beginning of the block.  The synchronization mechanism modifies the
  payload data to ensure that the signature does not appear by
  accident inside of the payload part (see,
  e.g.,~\cite{wijngaarden01}).

\end{itemize}

Let us give a succinct description of the synchronization codes that we
consider in this paper.  Our codes are designed for a
single control-and-data channel, with no access to a feedback link,
for the binary alphabet $\{ \Z, \U\}$. They provide hard
synchronization in a blockwise fashion where the blocks have a fixed
size before and after the encoding process, they insert a fixed
pattern of interspersed control bits, and they keep the original
payload bits unchanged.

\subsection{The Considered Synchronization Codes}

Let us first define the notion of \emph{phase}.  Given that we
consider transmission in a blockwise fashion, with fixed-size blocks,
each bit that is transmitted has a definite position relative to the
beginning of the block in which it appears.  This is the \emph{phase}
of the bit.  If one reads a stream of blocks starting at some unknown
position---not necessarily at the beginning of a block---one might be
interested in identifying the phase of the bits one reads.  Indeed,
identifying the phase allows one to recover synchronization.
Synchronization codes are intended to provide the means to identify
the phase.  We number the phases starting at zero.  The maximum phase
is the block size minus one.  More generally, the \emph{phase} of a
subsequence is the phase of its first bit.  If~$A$ is a sequence made
of one or more blocks and the block size is~$s$, then rotation~$A^i$
has phase~$i \bmod s$, for any~$i$.

A \emph{$(d, k, n)$-synchronization code}~$C$ is a blockwise code that
transforms a clear-text block of $d$~bits into a synchronized block of
$(d+k)$~bits.  $C$~is denoted by a sequence of symbols taken from the
alphabet $\{\Z, \U, \D\}$ where~\Z\ and~\U\ denote control bits and
\D\ denotes a data bit.  A $(d, k, n)$-synchronization code is
characterized by its three parameters:
\begin{itemize}
  \item $d$ is the number of data bits;
  \item $k$ is the number of control bits; and
  \item $n$ is the \emph{reliability}, i.e. the number of bits that
    need to be read, in the worst case, to identify the phase of the
    sequence.
\end{itemize}

\begin{definition}[Synchronization code]
\label{def::code}
$C$ is a \emph{$(d, k, n)$-synchronization code} iff it is a sequence
drawn from the alphabet $\{\Z, \U, \D\}$, it has length~$d+k$, it
contains exactly $d$~wildcards, and the subsequence $C^i[0..n-1]$ matches
the subsequence $C^j[0..n-1]$ only if $i = j \pmod{d+k}$.
\end{definition}

\begin{table}[t]
\caption{An $(8,10,9)$-synchronization code published
  by~\cite{dube11}.\label{table::example}}
\begin{center}
\begin{tabular}{%
       l@{\ \ }l@{\ \ }l@{\ \ }l@{\ \ }l@{\ \ }l@{\ \ }l@{\ \ }l@{\ \ }l@{\ }|%
  @{\ }l@{\ \ }l@{\ \ }l@{\ \ }l@{\ \ }l@{\ \ }l@{\ \ }l@{\ \ }l@{\ \ }l%
}
\D&\D&\D&\D&\Z&\Z&\Z&\U&\U&\D&\D&\D&\D&\Z&\U&\Z&\U&\U\\
\D&\D&\D&\Z&\Z&\Z&\U&\U&\D&\D&\D&\D&\Z&\U&\Z&\U&\U&\D\\
\D&\D&\Z&\Z&\Z&\U&\U&\D&\D&\D&\D&\Z&\U&\Z&\U&\U&\D&\D\\
\D&\Z&\Z&\Z&\U&\U&\D&\D&\D&\D&\Z&\U&\Z&\U&\U&\D&\D&\D\\
\Z&\Z&\Z&\U&\U&\D&\D&\D&\D&\Z&\U&\Z&\U&\U&\D&\D&\D&\D\\
\Z&\Z&\U&\U&\D&\D&\D&\D&\Z&\U&\Z&\U&\U&\D&\D&\D&\D&\Z\\
\Z&\U&\U&\D&\D&\D&\D&\Z&\U&\Z&\U&\U&\D&\D&\D&\D&\Z&\Z\\
\U&\U&\D&\D&\D&\D&\Z&\U&\Z&\U&\U&\D&\D&\D&\D&\Z&\Z&\Z\\
\U&\D&\D&\D&\D&\Z&\U&\Z&\U&\U&\D&\D&\D&\D&\Z&\Z&\Z&\U\\
\D&\D&\D&\D&\Z&\U&\Z&\U&\U&\D&\D&\D&\D&\Z&\Z&\Z&\U&\U\\
\D&\D&\D&\Z&\U&\Z&\U&\U&\D&\D&\D&\D&\Z&\Z&\Z&\U&\U&\D\\
\D&\D&\Z&\U&\Z&\U&\U&\D&\D&\D&\D&\Z&\Z&\Z&\U&\U&\D&\D\\
\D&\Z&\U&\Z&\U&\U&\D&\D&\D&\D&\Z&\Z&\Z&\U&\U&\D&\D&\D\\
\Z&\U&\Z&\U&\U&\D&\D&\D&\D&\Z&\Z&\Z&\U&\U&\D&\D&\D&\D\\
\U&\Z&\U&\U&\D&\D&\D&\D&\Z&\Z&\Z&\U&\U&\D&\D&\D&\D&\Z\\
\Z&\U&\U&\D&\D&\D&\D&\Z&\Z&\Z&\U&\U&\D&\D&\D&\D&\Z&\U\\
\U&\U&\D&\D&\D&\D&\Z&\Z&\Z&\U&\U&\D&\D&\D&\D&\Z&\U&\Z\\
\U&\D&\D&\D&\D&\Z&\Z&\Z&\U&\U&\D&\D&\D&\D&\Z&\U&\Z&\U
\end{tabular}
\end{center}
\end{table}

Table~\ref{table::example} presents an example of an
$(8,10,9)$-synchronization code and illustrates why it is
$9$-reliable.  The code appears on the first row.  The other rows are
mere rotations of the first row.  The reader may verify that, whenever
one selects two \emph{distinct} rows of the matrix, one is guaranteed
to find a mismatch between two control bits inside of the first
$9$~columns.  Note that one cannot rely on the data bits to cause bit
mismatches because data bits are completely out of the control of the
synchronization code.

We prove a property of $(d,k,n)$-synchronization codes that we will
reuse later.
\begin{lemma}
\label{lemma::bound}
In a $(d,k,n)$-synchronization code, the relation $k \leq 2^n-d$ holds.
\end{lemma}
\begin{proof}
Suppose that one starts reading a stream of bits from an unknown position in that stream. Since
the stream is encoded with a $n$-reliable synchronization code, it is sufficient to read
the next $n$ bits in the stream to find the phase in a sure
way. There are at most $d+k$ phases and these $n$ bits form one of the
$2^n$ possible sequence of $n$ bits. Each of these sequences can be
associated to at most one phase and all phases are covered by at least
one sequence which implies $d+k \leq 2^n$ and thus $k \leq 2^n -
d$. \qed
\end{proof}

\section{Compression by Substring Enumeration}
\label{sect:CSE}

Dub{\'e} and Beaudoin~\cite{dube10b} introduced a lossless data
compression technique called Compression by Substring Enumeration
(CSE).  CSE compresses by transmitting to the decompressor in a
half-explicit way the number~$C_w$ of occurrences of every possible
substring~$w$ of bits; i.e.\ for every $w \in \{ \Z, \U\}^*$.
Table~\ref{table:enumeration} displays the information that needs to
be transmitted, in a compressed form, for the
sequence~\texttt{01010000}.  For example: entry `\ddstrtimes{6}{0}'
indicates that $C_{\mathtt{0}} = 6$; entry `\ddstrtimes{1}{01010000}'
indicates that $C_{\mathtt{01010000}} = 1$; entry
`\ddstrtimes{8}{$\epsilon$}' means that the data is 8-bit long; and
the absence of entry for substring \texttt{11} means that
$C_{\mathtt{11}} = 0$.  CSE considers the data to be circular.
Consequently, we have $C_{\mathtt{00}} = 4$, not $C_{\mathtt{00}} =
3$, where the fourth occurrence is the one that wraps from the end to
the beginning of the data sequence.  We say that CSE describes the
numbers of occurrences in a \emph{half-explicit} way because the
numbers are not transmitted explicitly in a row-by-row fashion.
Rather, relations between the numbers of occurrences of different
substrings are exploited.

\begin{table}[t]
\caption{Number of occurrences of each substring of
  \texttt{01010000}.}
\label{table:enumeration}
\begin{lrbox}{\DDboxII}
\begin{tabular}{|c||llllllll|}
\hline
    \textbf{Length}
  &
    \multicolumn{8}{c|}{\textbf{Substrings}}
\\ \hline \hline
    0
  &
    \ddstrtimes{8}{$\epsilon$}
  &
  &
  &
  &
  &
  &
  &
\\ \hline
    1
  &
    \ddstrtimes{6}{0}
  &
  &
  &
  &
  &
  &
    \ddstrtimes{2}{1}
  &
\\ \hline
    2
  &
    \ddstrtimes{4}{00}
  &
  &
  &
  &
    \ddstrtimes{2}{01}
  &
  &
    \ddstrtimes{2}{10}
  &
\\ \hline
    3
  &
    \ddstrtimes{3}{000}
  &
  &
  &
    \ddstrtimes{1}{001}
  &
    \ddstrtimes{2}{010}
  &
  &
    \ddstrtimes{1}{100}
  &
    \ddstrtimes{1}{101}
\\ \hline
    4
  &
    \ddstrtimes{2}{0000}
  &
  &
    \ddstrtimes{1}{0001}
  &
    \ddstrtimes{1}{0010}
  &
    \ddstrtimes{1}{0100}
  &
    \ddstrtimes{1}{0101}
  &
    \ddstrtimes{1}{1000}
  &
    \ddstrtimes{1}{1010}
\\ \hline
    5
  &
    \ddstrtimes{1}{00000}
  &
    \ddstrtimes{1}{00001}
  &
    \ddstrtimes{1}{00010}
  &
    \ddstrtimes{1}{00101}
  &
    \ddstrtimes{1}{01000}
  &
    \ddstrtimes{1}{01010}
  &
    \ddstrtimes{1}{10000}
  &
    \ddstrtimes{1}{10100}
\\ \hline
    6
  &
    \ddstrtimes{1}{000001}
  &
    \ddstrtimes{1}{000010}
  &
    \ddstrtimes{1}{000101}
  &
    \ddstrtimes{1}{001010}
  &
    \ddstrtimes{1}{010000}
  &
    \ddstrtimes{1}{010100}
  &
    \ddstrtimes{1}{100000}
  &
    \ddstrtimes{1}{101000}
\\ \hline
    7
  &
    \ddstrtimes{1}{0000010}
  &
    \ddstrtimes{1}{0000101}
  &
    \ddstrtimes{1}{0001010}
  &
    \ddstrtimes{1}{0010100}
  &
    \ddstrtimes{1}{0100000}
  &
    \ddstrtimes{1}{0101000}
  &
    \ddstrtimes{1}{1000001}
  &
    \ddstrtimes{1}{1010000}
\\ \hline
    8
  &
    \ddstrtimes{1}{00000101}
  &
    \ddstrtimes{1}{00001010}
  &
    \ddstrtimes{1}{00010100}
  &
    \ddstrtimes{1}{00101000}
  &
    \ddstrtimes{1}{01000001}
  &
    \ddstrtimes{1}{01010000}
  &
    \ddstrtimes{1}{10000010}
  &
    \ddstrtimes{1}{10100000}
\\ \hline
\end{tabular}
\end{lrbox}
\begin{center}
  \resizebox{\textwidth}{!}{%
    \usebox{\DDboxII}%
  }%
\end{center}
\end{table}

As part of the processing of each substring~$w$, there is the crucial
step of predicting the value of the bits that surround all of the
occurrences of~$w$.  More precisely, for each~$w$ that occurs in the
original file, CSE predicts~$C_{\mathtt{0}w\mathtt{0}}$.  This
prediction is not made from scratch: it is based on the already known
values $C_{w}$, $C_{\mathtt{0}w}$, $C_{\mathtt{1}w}$,
$C_{w\mathtt{0}}$, and $C_{w\mathtt{1}}$.  Moreover, it is sufficient
for CSE to explicitly describe $C_{\mathtt{0}w\mathtt{0}}$ as the
three remaining unknown values $C_{\mathtt{0}w\mathtt{1}}$,
$C_{\mathtt{1}w\mathtt{0}}$, and $C_{\mathtt{1}w\mathtt{1}}$ can be
deduced from the known ones.  Even if it is still in its earliest
developments, CSE is already fairly competitive with more mature, well
known techniques such as dictionary-based
compression~\cite{ziv77,ziv78}, prediction by partial
matching~\cite{cleary84,cleary97}, the Burrows-Wheeler
transform~\cite{burrows94}, and compression using
anti-dictionaries~\cite{crochemore02}.  The reader might wonder: what
is the link between CSE and synchronization codes?

CSE
is an inherently bit-oriented technique.  Yet, all the benchmark files
(from the Calgary Corpus~\cite{witten87b}) that were used are
byte-oriented.  The conversion of the files into bit sequences is
trivial: each byte is simply viewed as a string of 8~bits.  However,
this change in the point of view of the data is not just cosmetic.  It
interferes with the prediction steps performed by CSE.  Indeed, let us
consider what happens when CSE predicts the neighbors of
substring~$\mathtt{100}$.  We illustrate the problem using the
following hypothetical excerpt from a benchmark file:
\[
    \ldots
  \ddgreymid
    \Z\B\Z\B\underline{\U\B\Z\B\Z}\B\Z\B\Z\B\U
  \ddgreymid
    \Z\B\U\B\Z\B\U\B\Z\B\U\B\U\B\underline{\U
  \ddgreymid
    \Z\B\Z}\B\U\B\Z\B\U\B\Z\B\U\B\Z
  \ddgreymid
    \ldots
\]
in which the two occurrences of~$\mathtt{100}$ have been underlined.
The light-grey vertical bars delimit the frontiers between the
original bytes but CSE is completely unaware of them.  We say that the
individual bits in some byte are located at phases~0, \ldots, 6, or~7,
depending on their position relative to a byte frontier.  For
instance, the first occurrence of~$\mathtt{100}$ in the excerpt is at
phase~2 while the second one is at phase~7.
Statistically, it is likely that
the~`$\mathtt{100}$' substrings located at a certain phase have
different neighboring bits than the~`$\mathtt{100}$'
substrings located at another phase.  Since CSE is unaware of the
phase, it makes predictions about the neighboring bits of all the
occurrences of~$\mathtt{100}$ at once, without regard to their phases.
It has been empirically observed that this phase unawareness incurs some
penalty for CSE.

Given this weakness of CSE, the authors had two choices: either adapt
CSE to effectively deal with bytes or use a trick to compensate for
phase unawareness.  The second option appeared to be easier and it
directly leads to the family of synchronization codes that we focus on
here~\cite{dube10b,dube11}.  Despite the fact that synchronizing
original data expands it, the codes in this family are exactly the
kind of patterns that CSE is very proficient at detecting and
exploiting.  Consequently, CSE does not suffer too much from
``learning" the newly inserted synchronization code.  Moreover, the
newly acquired artificial phase awareness leads to better compression
rates.  We refer the reader to Section~\ref{sect:experiments}.

\section{A Constraint Model}
\label{sect:constraint_model}

We present a model to compute a $(d,k,n)$-synchronization code.  The
model is built around $2k$ variables that are the positions and values
of the control bits in the synchronization code~$C$. Let $P_i$ be the
position of the $i^{th}$ control bit and $V_i$ be its corresponding
value.  Note that we will frequently have to compute some value modulo
$(d+k)$, especially after adding two values together or subtracting
one value from another.  From now on, we write~$x \oplus y$ as a
shorthand for $(x + y) \bmod (d + k)$ and $x \ominus y$ as a shorthand
for $(x - y) \bmod (d + k)$.

We have the following constraints.
\begin{align}
  P_i & \in \{0, 1, \ldots, d + k - 1\} & \forall \, 0 \leq i < k\\
  V_i & \in \{0, 1\} & \forall \, 0 \leq i < k\\
  \intertext{To break symmetries, we assume that the control bits are
    enumerated in the same order as they appear in the sequence.}
  P_{i-1} & < P_i & \forall \, 0 < i < k
\end{align}

\newcommand{\PA}{\ensuremath{P^A}}
\newcommand{\PB}{\ensuremath{P^B}}
\newcommand{\VA}{\ensuremath{V^A}}
\newcommand{\VB}{\ensuremath{V^B}}

From Definition~\ref{def::code}, we know that, for any $i \neq j
\pmod{d + k}$, there exists at least one control bit in $C^i[1..n]$
that does not match a control bit in $C^j[1..n]$.  Let the mismatching
bit in $C^i[1..n]$ be the $a^{th}$ control bit in the synchronization
code and the corresponding mismatching bit in $C^j[1..n]$ be the
$b^{th}$ control bit in the synchronization code. Since the $a^{th}$
bit and the $b^{th}$ bit are aligned when respectively shifted by $i$
and $j$ bits, we know that $P_a - P_b = i \ominus j$.  Moreover,
since the $a^{th}$ bit occurs in $C^i[0..n-1]$, its position in the
code $C$ is therefore in $\{i, i \oplus 1, \ldots, i \oplus (n-1)\}$.
Finally, these two bits exist for any $i \neq j \pmod{d + k}$. We
therefore create two variables for each $0 \leq i < j < d+k$.
\begin{align*}
  A_{i,j} & \in \{0, 1, \ldots, k - 1\} & \forall \, 0 \leq i < j < d + k \\
  B_{i,j} & \in \{0, 1, \ldots, k - 1\} & \forall \, 0 \leq i < j < d + k \\
  \intertext{Let $\PA_{i,j}$ be the position of the $a^{th}$ bit in the
    synchronization code.}
  \PA_{i,j} & \in \{i, i \oplus 1, \ldots, i \oplus (n - 1) \} &
  \forall \, 0 \leq i < j < d+k
  \intertext{Let $\PB_{i,j}$ be the position of the $b^{th}$ bit in the
    synchronization code.}
  \PB_{i,j} & \in \{j, j \oplus 1, \ldots, j \oplus (n - 1) \} &
  \forall \, 0 \leq i < j < d+k
  \intertext{Similarly, we define $\VA_{i,j}$ and $\VB_{i,j}$ to be the values of these control bits.}
  \VA_{i,j} & \in \{0, 1\} & \forall \, 0 \leq i < j < d+k \\
  \VB_{i,j} & \in \{0, 1\} & \forall \, 0 \leq i < j < d+k \\
  \intertext{These three constraints bind the variables $A_{i,j}$,
    $B_{i,j}$, $\PA_{i,j}$, and $\PB_{i,j}$ together. The constraint $\constraint{Element}([X_0, \ldots, X_{n-1}], Y, Z)$ ensures that $X_Y = Z$.}
  \element&([P_0, \ldots, P_{k-1}], A_{i,j}, \PA_{i,j}) & \forall \, 0 \leq i < j < d + k \\
  \element&([P_0, \ldots, P_{k-1}], B_{i,j}, \PB_{i,j}) & \forall \, 0 \leq i < j < d + k \\
  \PB_{i,j} & = \PA_{i,j} \oplus (j - i) & \forall \, 0 \leq i < j < d+k \\
  \intertext{Finally, since the $a^{th}$ bit does not match the
    $b^{th}$ bit, their characters must be different.}
  \element&([V_0, \ldots, V_{k-1}], A_{i,j}, \VA_{i,j}) & \forall \, 0 \leq i < j < d+k \\
  \element&([V_0, \ldots, V_{k-1}], B_{i,j}, \VB_{i,j}) & \forall \, 0 \leq i < j < d+k \\
  \VA_{i,j} & \neq \VB_{i,j}  & \forall \, 0 \leq i < j < d+k
\end{align*}

This model has a total of $O(k^2+d^2)$ variables and $O(k^2+d^2)$
constraints. The cardinality of the domains is bounded by
$\max(n,k)$ values.

We show in Section~\ref{sect:symmetries} how to further break symmetries.

\section{A Pseudo-Boolean Model}
\label{sect:pseudo_boolean_model}

\newcommand{\faii}[3]{\forall \; #2 \leq #1 \leq #3}
\newcommand{\faie}[3]{\forall \; #2 \leq #1 < #3}
\newcommand{\faei}[3]{\forall \; #2 < #1 \leq #3}
\newcommand{\faee}[3]{\forall \; #2 < #1 < #3}

The second model we present uses a different approach. Since the
synchronization code $c_0c_1\ldots c_{d+k-1}$ has $d+k$ characters
taken from $\{0, 1, \_\}$, we declare two binary variables for each of
these characters. The binary variable $K_i$ indicates whether the
$i^{th}$ character is a control bit. If the $i^{th}$ bit is a control
bit, the variable $V_i$ indicates whether it is a zero or a one.
\begin{align*}
  K_i & \in \{0, 1\} & \faie{i}{0}{d+k} \\
  V_i & \in \{0, 1\} & \faie{i}{0}{d+k}
\intertext{We fix the number of control bits to be $k$.}
\sum_{i=0}^{d+k-1} K_i & = k
\end{align*}

We declare a new binary variable $Y_i^g$ for $0 \leq i < d+k$ and $1 \leq
g < d+k$. When this variable is equal to one, it implies that the bits
$c_i$ and $c_{i \oplus g}$ are distinct control bits, i.e. $Y_i^g =1 \Rightarrow K_i = K_{i \oplus g} = 1 \land V_i \neq V_{i \oplus g}$. We
encode this implication with these constraints.
\begin{align*}
  Y_i^g & \in \{0, 1\} & \faie{i,j}{0}{d+k} \\
  Y_i^g & \leq K_i & \faie{i,j}{0}{d+k} \\
  Y_i^g & \leq K_{i \oplus g} & \faie{i,j}{0}{d+k} \\
  (1 - Y_i^g) + V_i + V_{i \oplus g} & \geq 1 & \faie{i,j}{0}{d+k} \\
  (1 - Y_i^g) + (1 - V_i) + (1 - V_{i \oplus g})  & \geq 1  & \faie{i,j}{0}{d+k} \\
  \intertext{Referring to Definition~\ref{def::code}, we make sure
    that there is a control bit in $C^i[0..n-1]$ that differs from a
    control bit in $C^j[0..n-1]$. This is ensured by this constraint.}
  \sum_{p \in \{i \oplus a \mid a = 0..n-1\}} Y_p^{j-i} & \geq 1 & \forall \, 0 \leq i < j < d+k\\
  \intertext{To break symmetries, we force the variable $V_i$ to be
    assigned value~$0$ whenever $c_i$ is not a control bit.}
  K_i & \geq V_i& \forall \, 0 \leq i < d+k
\end{align*}

As for the constraint model of
Section~\ref{sect:pseudo_boolean_model}, this model has a total of
$O(k^2+d^2)$ variables and $O(k^2+d^2)$ constraints. However, all
domains contain two values.

We show in the next section how to break additional symmetries.

\section{Symmetries}
\label{sect:symmetries}

For any synchronization code $C$, there exist several other codes that
are equivalent. For instance, all zeros in $C$ can be changed to ones
and all ones changed to zeros. This produces a valid code.  If $C$ is
a synchronization code, the rotation $C^i$ is also a synchronization
code for any integer $i$.  Finally, reverting the sequence $C$ such
that the first character becomes last and the last character becomes
first also produces a valid code.

In order to break symmetries in both models, we force the first two
characters of the code to be 0 and 1. Indeed, any synchronization code
must have two adjacent bits of different value in order for the
subsequences $C^0[0..n-1]$ and $C^1[0..n-1]$ not to match. These bits
could be anywhere in the sequence $C$, but after applying a rotation,
one can always make them appear at the first two positions of the
sequence.

We also force the number of control bits set to zero to be no fewer
than the number of control bits set to one.

In the constraint model of Section~\ref{sect:constraint_model}, we add
these constraints.
\begin{align*}
  P_0 & = 0 & P_1 & = 1 & V_0 & = 0 & V_1 & = 1 &
  \sum_{i=0}^{k-1}V_i & \leq \frac{k}{2} \\
  \intertext{In the pseudo-Boolean model of
    Section~\ref{sect:pseudo_boolean_model}, we add these
    constraints}
  K_0 & = 1 & K_1 & = 1 & V_0 & = 0 & V_1 & = 1 & \sum_{i=0}^{d+k-1}
  V_i & \leq \frac{k}{2}
\end{align*}

\section{Experiments}
\label{sect:experiments}

We implemented the constraint model using Gecode 3.7.3 and solved the
pseudo-Boolean model using the solver MiniSat+~\cite{minisat}. All experiments were
conducted on a 2.4 GHz Intel Core i7 machine, with 4 Go, 1333 MHz DDR3
RAM.

We tried Gecode using different predefined branching heuristics. The
most efficient method was to choose the variable with the smallest
domain size divided by the weighted degree of the variable ({\tt
  INT\_VAR\_SIZE\_AFC\_MIN})~\cite{DBLP:conf/ecai/BoussemartHLS04} --- a variable ordering heuristic that gives a higher priority to variables on which failures occur more often.
Gecode was able to solve small instances such as $d = n = 8, k=15$ in
5 minutes and 18 seconds. However, the solver could not prove that the
instance $d = n = 8, k = 14$ is unsatisfiable even after a month of
computation.

\newcommand{\figurescale}{1.45}
\begin{figure}[t]
\begin{center}
\begin{tabular}{@{}cc@{}}
\scalebox{\figurescale}{\includegraphics{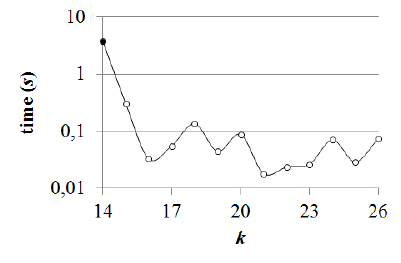}} &
\scalebox{\figurescale}{\includegraphics{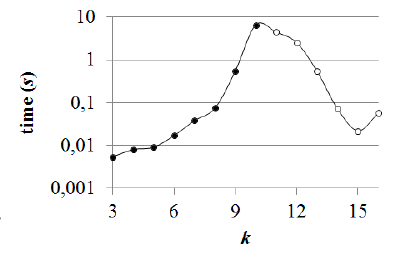}} \\
$d=n=8$ &
$d=n=16$ \\
\scalebox{\figurescale}{\includegraphics{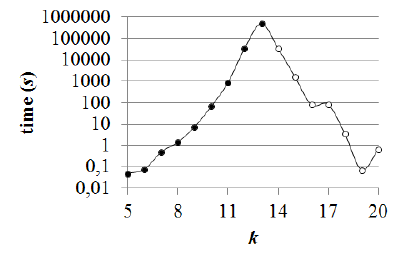}} &
\scalebox{\figurescale}{\includegraphics{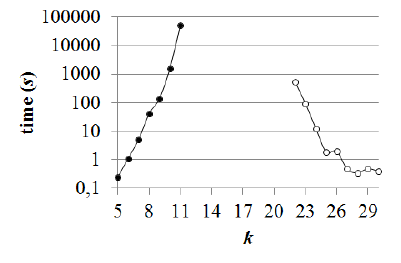}} \\
$d=n=32$ &
$d=n=64$
\end{tabular}
\caption{Time required by MiniSat+ to find a $(d,k,n)$-synchronization
  code or prove that no such a code exists.  The number of data
  bits~$d$ and the reliability~$n$ are fixed, only the number~$k$ of
  control bits vary. A black circle represents an instance for which
  there is no possible $(d,k,n)$-synchronization code, and a white
  circle represents that such code exists.
\label{fig::results}}
\end{center}
\end{figure}


We tried to solve the pseudo-Boolean model with the solver MiniSat+
using the default configuration. The solver solves the instance
$d=n=8, k=15$ in 0.31 second and proves the unsatisfiability of the
instance $d=n=8,k=14$ in 3.77 seconds. We also tried the
pseudo-Boolean solver Sat4j~\cite{sat4j} but MiniSat+ turned out to be more
competitive on most instances\footnote{For instance, we showed there is no $(32,13,32)$-synchronization code after about 6 days using MiniSat+ while Sat4j did not return any answer after 30 days.}. We conclude that the combination of the
pseudo-Boolean model and the solver MiniSat+ is more efficient than
the constraint model.


We first investigate instances with $d = n$ since those are the codes
that are the most used.  A machine that manipulates data divided into
words of $d$ bits is more likely to synchronize the stream of bits
after reading $n=d$ bits.  Generally, the fewer the synchronization bits,
the better. In some cases, a code with additional synchronization bits could
perform better for compression purposes.  

Figure~\ref{fig::results} presents the time needed for MiniSat+ to
find a synchronization code or to prove that no such code exists.  For
$(d,k,n)$-synchronization codes with $d=n=8$, $d=n=16$ and $d=n=32$
data bits, we proved that the smallest number of control bits is $k =
15$, $k = 11$ and $k = 14$ respectively. For
$(64,k,64)$-synchronization codes, we found a code with $k=22$ control
bits and proved that no such codes exists with $k \leq 11$ bits. We
previously solved most of the same instances with the solver Sat4j and
proved there is a $(64,21,64)$-synchronization code. However, we have
not yet proved there is one using MiniSat+, so the result does not
show in Figure~\ref{fig::results}. Whether there exists a
$(64,k,64)$-synchronization code for $11 < k < 21$ is an open
question. MiniSat+ did not succeed to find such codes after 10 hours
of computation.

\begin{table}[t]
\caption{Smallest values of $k$ relative to $n$ and $d$.  An entry
  with~$\infty$ indicates that the solver proved that no such code
  exists. A blank entry indicates the incapacity of the solver to
  prove or disprove the existence of the code withing 18000 seconds of
  computation.}
\label{table::dkn}
\begin{center}
\begin{tabularx}{\columnwidth}{XX|XXXXXXXXXXXX}
\hline
\hline
\multicolumn{2}{c}{} & \multicolumn{11}{c}{$n$} \\
\multicolumn{2}{c}{} & 2 & 3 & 4 & 5 & 6 & 7 & 8 & 9 & 10 & 11 & 12 & 13  \\
\cline{3-14}
\multirow{7}{*}{$d$} & 2 & {$\infty$} & $\infty$ & 6 & 3 & 3 & 3 & 3 & 3 & 3 & 3 & 3 & 3 \\
& 3 & $\infty$ & {$\infty$} & $\infty$ & 12 & 6 & 4 & 4 & 4 & 4 & 4 & 4 & 4 \\ 
& 4 & $\infty$ & $\infty$ & {$\infty$} & $\infty$ & 8 & 7 & 4 & 4 & 4 & 4 & 4 & 4 \\ 
& 5 & $\infty$ & $\infty$ & $\infty$ & {$\infty$} &  & 9 & 7 & 4 & 4 & 4 & 4 & 4 \\ 
& 6 & $\infty$ & $\infty$ & $\infty$ & $\infty$ &  & 8 & 8 & 8 & 7 & 5 & 5 & 5 \\ 
& 7 & $\infty$ & $\infty$ & $\infty$ & $\infty$ &  & {18} & 13 & 9 & 8 & 8 & 5 & 5 \\ 
& 8 & $\infty$ & $\infty$ & $\infty$ & $\infty$ &  & 20 & {15} & 10 & 10 & 8 & 8 & 5 \\ 
 \hline \hline

\end{tabularx}
\end{center}
\end{table}

As a second experiment, we try to find the codes with the fewest
number of control bits when $d \neq n$. We fix the values of $d$ and
$n$ and let MiniSat+ solve the instances with $k=1, 2, 3,\ldots$ until
the solver finds a
$(d,k,n)$-synchronization code. If a $(d,k,n)$-synchronization code is
found, the value of $k$ is written in Table~\ref{table::dkn}. If no
$(d,k,n)$-synchronization code exists for $k \leq 2^n - d$, then
Lemma~\ref{lemma::bound} ensures that no synchronization code exists
for any value of $k$ and therefore $\infty$ is written in
Table~\ref{table::dkn}. If the solver does not find a code nor prove that it
does not exists after 18000
seconds, the entry is left blank.

%
%
%
%
%
%
%



Table~\ref{table::dkn} shows that our model is efficient with
different instances of $d$, $k$, and $n$.  It also validates
previously published results obtained using a different methodology
\cite{dube11} and extends it to a wider range of values of~$d$
and~$n$.

\begin{table}[t]
\caption{Synchronization codes used to boost CSE in previous
  research~\cite{dube11}.}
\label{table:DCC11codes}
\begin{lrbox}{\ddboxI}
  \begin{tabular}[t]{cc|l}
      $n$
    &
      $k$
    &
      \multicolumn{1}{c}{\textbf{Synchronization Scheme}}
  \\ \hline
      ---
    &
      $0$
    &
      \oA \oA \oA \oA \oA \oA \oA \oA
  \\
      ---
    &
      $1$
    &
      \oA \oA \oA \oA \oA \oA \oA \oA \aZ
  \\
      ---
    &
      $2$
    &
      \oA \oA \oA \oA \oA \oA \oA \oA \aZ \aU
  \\
      ---
    &
      $3$
    &
      \oA \oA \oA \oA \oA \oA \oA \oA \aZ \aU \aU
  \\
      ---
    &
      $4$
    &
      \oA \oA \oA \oA \oA \oA \oA \oA \aZ \aU \aU \aU
  \\
      $13$
    &
      $5$
    &
      \oA \oA \oA \oA \oA \oA \aZ \oA \oA \aZ \aU \aU \aU
  \\
      $12$
    &
      $8$
    &
      \oA \oA \oA \aZ \oA \oA \aU \aZ \aZ \oA \oA \aU \oA \aU \aU \aZ
  \\
      $11$
    &
      $8$
    &
      \oA \oA \oA \aZ \oA \aZ \oA \oA \aU \aU \aZ \oA \oA \aU \aU \aZ
  \\
      $10$
    &
      $10$
    &
      \oA \oA \oA \oA \aZ \oA \aZ \aU \aU \oA \oA \oA \aZ \aU \aZ \aZ \aU \aU
  \\
      $9$
    &
      $10$
    &
      \oA \oA \oA \oA \aZ \aZ \aZ \aU \aU \oA \oA \oA \oA \aZ \aU \aZ \aU \aU
  \\
      $8$
    &
      $15$
    &
      \oA \oA \oA \aZ \aZ \aZ \aU \aZ \oA \oA \oA \aU
      \aU \aU \aZ \aU \oA \oA \aU \aU \aZ \aZ \aU
  \\
      $7$
    &
      $20$
    &
      \aU \aU \aZ \aU \oA \aU \oA \aU \aU \aZ \aZ \oA \aZ \oA
      \aZ \aU \aZ \aZ \oA \aZ \oA \aU \aU \aZ \aZ \oA \aU \oA
  \end{tabular}
\end{lrbox}
\begin{center}
    \usebox{\ddboxI}%
\end{center}
\end{table}

\begin{table}[t]
\caption{Compression performance obtained while boosting CSE using the
  synchronization of Table~\ref{table:DCC11codes}, as presented
  in~\cite{dube11}.  Measurements are in bits per character.}
\label{table:DCC11bpcs}
\begin{lrbox}{\ddboxI}
\begin{tabular}{|@{\ }l@{\ }||@{\ }c@{\ }|@{\ }c@{\ }%
                |@{\ }c@{\ }|@{\ }c@{\ }|@{\ }c@{\ }%
                |@{\ }c@{\ }|@{\ }c@{\ }|@{\ }c@{\ }%
                |@{\ }c@{\ }|@{\ }c@{\ }|@{\ }c@{\ }%
                |@{\ }c@{\ }|@{\ }c@{\ }|@{\ }c@{\ }%
                |@{\ }c@{\ }|}
\ddTR{\resizebox{0.8\width}{\height}{Bits$/$Car.}}
             {BWT} {PPM} {Anti}{ k=0}{ k=1}{ k=2}{ k=3}
     { k=4}{n=13}{n=12}{n=11}{n=10}{ n=9}{ n=8}{ n=7}
\ddRR{bib   }{2.07}{1.91}{2.56}{1.98}{1.95}{1.92}{1.92}
     {1.91}{1.90}{1.89}{1.89}{1.89}{1.89}{1.88}{1.88}
\ddRR{book1 }{2.49}{2.40}{3.08}{2.27}{2.26}{2.25}{2.25}
     {2.25}{2.25}{2.25}{2.25}{2.25}{2.25}{2.29}{2.33}
\ddRR{book2 }{2.13}{2.02}{2.81}{1.98}{1.96}{1.95}{1.95}
     {1.94}{1.94}{1.93}{1.93}{1.93}{1.93}{1.93}{1.95}
\ddRR{geo   }{4.45}{4.83}{6.22}{5.35}{5.21}{4.98}{4.81}
     {4.70}{4.63}{4.58}{4.59}{4.58}{4.58}{4.58}{4.57}
\ddRR{news  }{2.59}{2.42}{3.42}{2.52}{2.49}{2.46}{2.45}
     {2.44}{2.43}{2.43}{2.43}{2.43}{2.43}{2.42}{2.42}
\ddRR{obj1  }{3.98}{4.00}{4.87}{4.46}{4.53}{4.43}{4.32}
     {4.24}{4.17}{4.03}{4.05}{4.02}{4.01}{4.00}{3.99}
\ddRR{obj2  }{2.64}{2.43}{3.61}{2.71}{2.69}{2.59}{2.53}
     {2.49}{2.47}{2.45}{2.46}{2.45}{2.45}{2.45}{2.44}
\ddRR{paper1}{2.55}{2.37}{3.17}{2.54}{2.51}{2.48}{2.47}
     {2.46}{2.44}{2.41}{2.41}{2.41}{2.41}{2.41}{2.41}
\ddRR{paper2}{2.51}{2.36}{3.14}{2.41}{2.39}{2.38}{2.38}
     {2.37}{2.36}{2.35}{2.35}{2.34}{2.35}{2.34}{2.34}
\ddRR{paper3}{---} {---} {---} {2.73}{2.70}{2.69}{2.68}
     {2.67}{2.65}{2.63}{2.63}{2.63}{2.63}{2.63}{2.63}
\ddRR{paper4}{---} {---} {---} {3.20}{3.16}{3.13}{3.13}
     {3.10}{3.07}{3.02}{3.02}{3.02}{3.02}{3.01}{3.01}
\ddRR{paper5}{---} {---} {---} {3.33}{3.29}{3.27}{3.24}
     {3.22}{3.19}{3.12}{3.13}{3.12}{3.12}{3.11}{3.10}
\ddRR{paper6}{---} {---} {---} {2.65}{2.61}{2.58}{2.56}
     {2.55}{2.52}{2.50}{2.50}{2.49}{2.50}{2.49}{2.49}
\ddRR{pic   }{0.83}{0.85}{1.09}{0.77}{0.84}{0.82}{0.82}
     {0.82}{0.81}{0.81}{0.81}{0.81}{0.81}{0.81}{0.81}
\ddRR{progc }{2.58}{2.40}{3.18}{2.60}{2.58}{2.54}{2.52}
     {2.50}{2.48}{2.44}{2.44}{2.44}{2.44}{2.44}{2.43}
\ddRR{progl }{1.80}{1.67}{2.24}{1.71}{1.70}{1.69}{1.68}
     {1.67}{1.66}{1.65}{1.65}{1.65}{1.65}{1.64}{1.64}
\ddRR{progp }{1.79}{1.62}{2.27}{1.78}{1.76}{1.73}{1.71}
     {1.70}{1.68}{1.66}{1.66}{1.66}{1.66}{1.66}{1.65}
\ddRR{trans }{1.57}{1.45}{1.94}{1.60}{1.58}{1.53}{1.52}
     {1.50}{1.48}{1.47}{1.47}{1.47}{1.47}{1.47}{1.46}
\end{tabular}
\end{lrbox}
\begin{center}
  \resizebox{\textwidth}{!}{%
    \usebox{\ddboxI}%
  }
\end{center}
\end{table}

Even though this paper addresses the problem of finding synchronization
codes, it remains interesting to have an illustration of the
effectiveness of the synchronization codes on CSE's performance.
Tables~\ref{table:DCC11codes} and~\ref{table:DCC11bpcs} show
synchronization codes and compression performance measurements,
respectively, presented in previous research by
Dub\'e~\cite{dube11}.\footnote{That previous research has been carried
  on prior to that presented in this very paper.  The synchronization
  codes that were used have been obtained using various means, some
  manual, other computational.}  In Table~\ref{table:DCC11codes}, the $n$ and $k$ parameters are indicated for each code. Certain codes have no associated $n$; these are unreliable (soft) codes.  Apart from the first three columns, the measurements presented in Table~\ref{table:DCC11bpcs} correspond
to the use of the codes of Table~\ref{table:DCC11codes}.  The first three columns present
the compression performance of three well-known compression techniques: the
Burrows-Wheeler transform, prediction by partial matching, and compression using
antidictionaries.  The files in the
benchmark come from the Calgary Corpus~\cite{witten87b}. We refer the reader to the original paper for a complete description of these experiments~\cite{dube11}. One might notice that most of the boosting effect is obtained as soon as the synchronization code is reliable.  Still, there exist many
applications other than CSE that may benefit from very strong codes.

\section{Conclusion}

We presented the combinatorial problem of finding a
$(d,k,n)$-synchronization code for CSE compression algorithms.  We
presented two models for solving the problem: a constraint programming
model and a pseudo-Boolean model. The pseudo-Boolean model, when
solved with MiniSat+, turned out to be more efficient. We were able to
produce $(64,k,64)$-synchronization codes for large instances
requiring as few as 21~control bits.  We proved that for
$(7,k,7)$-synchroni\-zation codes to exist, we need $k \geq 18$,
for $(8,k,8)$-codes, we need $k \geq 15$, for
$(16,k,16)$-codes, we need $k \geq 11$,
and for $(32,k,32)$-synchronization codes, we need $k \geq 14$.
These bounds are tight.
We also computed the minimum number of control bits required for
instances with $2 \leq d \leq 8$ and $2 \leq n \leq 13$.

\subsubsection*{Acknowledgments}
The authors would like to thank Julia Gustavsen and Jiye Li for their
support while writing this paper. This work is funded by the Natural
Sciences and Engineering Research Council of Canada.

\bibliographystyle{splncs}
\bibliography{references}

\begin{thebibliography}{10}

\bibitem{bruyere97}
Bruyère, V.:
\newblock A completion algorithm for codes with bounded synchronization delay.
\newblock In: Proceedings of the International Colloquium on Automata,
  Languages and Programming, Bologna, Italy (1997)  87--97

\bibitem{do02}
Do, L.V., Litovsky, I.:
\newblock On a family of codes with bounded deciphering delay.
\newblock In: Proceedings of the International Conference on Developments in
  Language Theory, Kyoto, Japan (2002)  369--380

\bibitem{golomb65}
Golomb, S.W., Gordon, B.:
\newblock Codes with bounded synchronization delay.
\newblock Information and Control \textbf{8} (1965)  355--372

\bibitem{scholtz66}
Scholtz, R.A.:
\newblock Codes with synchronization capability.
\newblock IEEE Transactions on Information Theory \textbf{12} (1966)  135--142

\bibitem{wijngaarden01}
van Wijngaarden, A.J., Morita, H.:
\newblock Partial-prefix synchronizable codes.
\newblock IEEE Transactions on Information Theory \textbf{47} (2001)
  1839--1848

\bibitem{stiffler71}
Stiffler, J.J.:
\newblock Theory of synchronous communications.
\newblock Prentice-Hall (1971)

\bibitem{dube10b}
Dubé, D.:
\newblock Using synchronization bits to boost compression by substring
  enumeration.
\newblock In: Proceedings of the International Symposium on Information Theory
  and its Applications, Taichung, Taiwan (2010)

\bibitem{dube11}
Dubé, D.:
\newblock On the use of stronger synchronization to boost compression by
  substring enumeration.
\newblock In: Proceedings of the Data Compression Conference, Snowbird, Utah,
  USA (2011)

\bibitem{dube10a}
Dubé, D., Beaudoin, V.:
\newblock Lossless data compression via substring enumeration.
\newblock In: Proceedings of the Data Compression Conference, Snowbird, Utah,
  USA (2010)  229--238

\bibitem{tcp89}
R.~Braden, E.:
\newblock {RFC} 1122 (1989) \texttt{http://tools.ietf.org/html/rfc1122}.

\bibitem{ziv77}
Ziv, J., Lempel, A.:
\newblock A universal algorithm for sequential data compression.
\newblock IEEE Trans.\ on Information Theory \textbf{23} (1977)  337--342

\bibitem{ziv78}
Ziv, J., Lempel, A.:
\newblock Compression of individual sequences via variable-rate coding.
\newblock IEEE Trans.\ on Information Theory \textbf{24} (1978)  530--536

\bibitem{cleary84}
Cleary, J.G., Witten, I.H.:
\newblock Data compression using adaptive coding and partial string matching.
\newblock IEEE Trans.\ on Communications \textbf{32} (1984)  396--402

\bibitem{cleary97}
Cleary, J.G., Teahan, W.J.:
\newblock Unbounded length contexts for {PPM}.
\newblock The Computer Journal \textbf{40} (1997)  67--75

\bibitem{burrows94}
Burrows, M., Wheeler, D.:
\newblock A block sorting lossless data compression algorithm.
\newblock Technical Report 124, Digital Equipment Corporation (1994)

\bibitem{crochemore02}
Crochemore, M., Navarro, G.:
\newblock Improved antidictionary based compression.
\newblock In: Proceedings of the International Conference of the Chilean
  Computer Science Society. (2002)  7--13

\bibitem{witten87b}
Witten, I., Bell, T., Cleary, J.:
\newblock The {Calgary} corpus (1987)
  \mbox{}\newline\texttt{ftp://ftp.cpsc.ucalgary.ca/pub/projects/text.compression.corpus}.

\bibitem{minisat}
E{\'e}n, N., S{\"o}rensson, N.:
\newblock Translating pseudo-boolean constraints into {SAT}.
\newblock JSAT \textbf{2} (2006)  1--26

\bibitem{DBLP:conf/ecai/BoussemartHLS04}
Boussemart, F., Hemery, F., Lecoutre, C., Sais, L.:
\newblock Boosting systematic search by weighting constraints.
\newblock In: ECAI. (2004)  146--150

\bibitem{sat4j}
Berre, D.L., Parrain, A.:
\newblock The sat4j library, release 2.2.
\newblock Journal on Satisfiability, Boolean Modeling and Computation
  \textbf{7} (2010)  59--64

\end{thebibliography}

\end{document}